\documentclass{article}
\usepackage[utf8]{inputenc}
\usepackage{mathrsfs}
\usepackage{amsmath}
\usepackage{amssymb}
\usepackage{amsthm}
\usepackage{authblk}
\usepackage{tikz}
\usetikzlibrary{automata, positioning}
\usepackage{float}
\usepackage{tabularx,lipsum,environ}
\usepackage{algorithmic}

\newtheorem{fact}{Fact}
\newtheorem{theorem}{Theorem}
\newtheorem{lemma}{Lemma}
\newtheorem{corollary}{Corollary}

\title{On synchronization of partial automata}
\author{Jakub Ruszil}
\affil{Jagiellonian University}
\date{June 2020}

\begin{document}

\maketitle

\begin{abstract}
A goal of this paper is to introduce the new construction of an automaton with shortest synchronizing word of length $O(d^{\frac{n}{d}})$, where $d \in \mathbb{N}$ and $n$ is the number of states for that automaton. Additionally we introduce new transformation from any synchronizable DFA or carefully synchronizable PFA of $n$ states to carefully synchronizable PFA of $d \cdot n$ states with shortest synchronizing word of length $\Omega(d^{\frac{n}{d}})$.
\end{abstract}

\section{Synchronization of partial automata}
\emph{Partial finite automaton} (PFA) is an ordered tuple $\mathcal{A} = (\Sigma, Q, \delta)$ where $\Sigma$ is a set of letters, $Q$ is a set of states and $\delta:{Q \times \Sigma}\rightarrow{Q}$ is a transition function, not everywhere defined. For $\emph{w} \in \Sigma^\ast$ and $\emph{q} \in Q$ we define $\delta(\emph{q},\emph{w})$ inductively as $\delta(\emph{q},\epsilon) = q$ and $\delta(\emph{q},\emph{aw}) = \delta(\delta(\emph{q},\emph{a}), \emph{w})$ for $a \in \Sigma$ where $\epsilon$ is the empty word and $\delta(\emph{q}, \emph{a})$ is defined. A word $\emph{w} \in \Sigma^\ast$ is called \emph{carefully synchronizing} if there exists $\overline{q} \in Q$ such that for every $\emph{q} \in Q$, $\delta(\emph{q}, \emph{w}) = \overline{q}$ and all transitions are defined. A PFA is called \emph{carefully synchronizing} if it admits any carefully synchronizing word. Carefully synchronizing automaton $\mathcal{A}_{car}$ is depicted on the Fig. 1 and its shortest carefully synchronizing word $w_{car}$ is $abca^3b^2ca$ what can be easily checked via power automaton construction analogous to the deterministic automata construction. The only difference is that we define only transitions we can define.
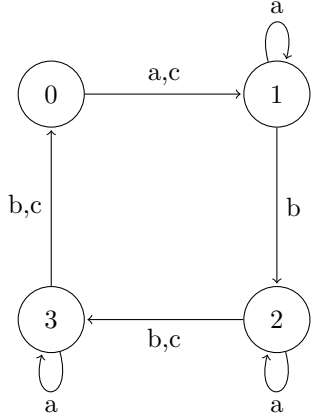
\begin{figure}[H]
    \centering
    \begin{tikzpicture}[shorten >=1pt,node distance=3cm,on     grid,auto] 
        \node[state] (0)   {$0$}; 
        \node[state] (1) [right=of 0] {$1$}; 
        \node[state] (2) [below=of 1] {$2$}; 
        \node[state] (3) [left=of 2] {$3$};
        \path[->] 
        (0) edge node {a,c} (1)
        (1) edge [loop above] node {a} ()
            edge node {b} (2)
        (2) edge [loop below] node {a} ()
            edge node {b,c} (3)
        (3) edge [loop below] node {a} ()
            edge node {b,c} (0)
;
        \end{tikzpicture}
    \caption{A carefully synchronizing $\mathcal{A}_{car}$}
\end{figure}
The concept of careful synchronization of PFA is a generalization of idea of synchronization for deterministic finite automata (DFA) with transition functions defined everywhere. The problem of estimating the value of $d(n)$ was considered first by Ito and Shikishima-Tsuji in [6-7] and later by Martyugin [8]. Ito and Shikishima-Tsuji proved that $2^{\frac{n}{2}} + 1 \leq d(n) \leq 2^n - 2^{n-2} - 1$ and Martyugin improved the lower bound with the construction of automata of length $O(3^{\frac{n}{3}})$. The best known upper bound for $d(n)$ is $O(n^2 \cdot 4^{\frac{n}{3}})$ due to [9].\newline Let $\mathcal{L}_n = \{\mathcal{A} = (\Sigma, Q, \delta): \mathcal{A}\;is\;carefully\;synchronizing\;and\;|Q| = n\}$. We define $d(\mathcal{A}) = min\{|w|:w\; is\; carefully\; synchronizing\:word\; for\; \mathcal{A}\}$ and $d(n) = max\{d(\mathcal{A}) : \mathcal{A} \in \mathcal{L}_n\}$. It can be easily verified from Fig 1. that the \v{C}ern\'y Conjecture is not true for PFAs, since $|w_{car}| = 10 > (4-1)^2 = 9$. We also recall following important facts.

\begin{fact}
Let $\mathcal{A}$ be a PFA and $\mathcal{P(A)}$ be its power automaton. Then $\mathcal{A}$ is synchronizing if and only if for some state $q \in Q$ there exists a labelled path in $\mathcal{P(A)}$ from $Q$ to $\{q\}$. The shortest synchronizing word for $\mathcal{A}$ corresponds to the shortest labelled path in $\mathcal{P(A)}$ as above. 

\end{fact}

\begin{fact}
If automaton $\mathcal{A}$ is carefully synchronizing then there exists $a' \in \Sigma$ such that transition under $a'$ is defined for all states and $q_1', q_2' \in Q$ such, that $\delta(q_1',a') = \delta(q_2',a')$.
\end{fact}

Now we are ready to give an example of a PFA with shortest carefully synchronizing word of length $O(d^{n/d})$ for all $d \in \mathbb{N}$.

\section{Automata with long shortest carefully synchronizing words}
This section includes construction of an automaton for which the shortest carefully synchronizing word is of exponential length. \newline
Let $d \in \mathbb{N}$, $d > 1$, $a_1, ..., a_k \in \{0, ..., d-1\}$ and $r = \sum\limits_{i=1}^k a_i \cdot d^{i-1}$. We understand $(a_k, ..., a_1)_d$ as base $d$ representation of $r$. 
Let $n = d \cdot k$, $k \in \mathbb{N}$. We define automaton $\mathcal{A}_d(n) = (\Sigma, Q, \delta)$ as follows:
\begin{itemize}
  \item $\Sigma = \{a, b_1, b_2, ... b_k, c_k, c_{k-1}, ... , c_2\}$
  \item $Q_i = \{q_0^i, q_1^i, ... ,  q_{d-1}^i\} $
  \item $Q = \bigcup\limits_{i = 1}^k Q_i$
\end{itemize} 
Let $i \in \{1,..., k\}$, $l \in \mathbb{N}$. we define partial transition function $\delta:{Q \times \Sigma}\rightarrow{Q}$ for $\mathcal{A}_d(n)$ as:
\begin{itemize}
\item     $\delta(q_j^i, a) = q_0^i$, $j \in \{0,1, ... , d-1\}$
\item  $\delta(q_{j-1}^i, b_i) = q_{j}^i$, $j \in \{1,2, ... , d-1\}$
\item $\delta(q_j^i, b_l) = q_j^i$, $i > l$
\item $\delta(q_{d-1}^i, b_l) = q_0^i$, $i < l$
\item $\delta(q_{d-1}^i, c_i) = q_0^{i-1}$
\item $\delta(q_{d-1}^i, c_l) = q_0^i$, $i < l$
\end{itemize}
Let us remark some facts about the construction useful for further proofs.
\begin{fact}
$\delta(q_j^i, b_l)$ is not defined when $i < l$ and $j \in \{0, ..., d-2\}$.
\end{fact}

\begin{fact}
$\delta(q_{d-1}^i, b_i)$ is not defined.
\end{fact}

\begin{fact}
$\delta(q_j^i, c_i)$ is not defined when  $j \in \{0, ..., d-2\}$.
\end{fact}

It is worth noticing that only transitions on letters $c_i$ and $a$, $i \in \{1, ..., k\}$ join two states together and only letter $a$ is defined for all states.\newline
Let $m \in \mathbb{N}$ and $r = (j_m, ..., j_1)_d$. We also define $Q_r^m \subset Q$ such, that:
\begin{itemize}
    \item $j_i$ for $i \in \{1,...,m\}$ corresponds to lower index of $q_j^i \in Q_r^m$
    \item $|Q_r^m| = m$
    \item $|\{q_0^i,q_1^i, ... ,q_{d-1}^i\} \cap Q_r^m| = 1$, $i \in \{1, ... , m\}$
\end{itemize}
In other words each $Q_r^m$ corresponds to m-digit base $d$ representation of $r$. For example if $d=3$, then $Q_{10}^4 = \{q_0^4, q_1^3,q_0^2,q_1^1\}$. Finally we define inductively word $w_i \in \Sigma^i$ as:\newline\newline
$w_i =
\left\{\begin{array}{ll}
		\epsilon  & \mbox{if } i = 0 \\
		(w_{i-1}b_i)^{d-1}w_{i-1} & \mbox{if } i > 0
	\end{array}
\right.
$\newline\newline
Now we are ready to formulate first lemma of this section.

\begin{lemma}
Let $n = d \cdot k$ and $\mathcal{A}_d(n)$ be defined as above. For every $i \in \{1, ..., k\}$ there exists a path $(Q_0^i, Q_1^i,Q_2^i, ... ,Q_{d^i - 1}^i )$ in $\mathcal{P}(\mathcal{A}(n))$ and its transitions are labelled with consecutive letters of word $w_i$.
\end{lemma}

\begin{proof}
The result follows by induction on i.\newline
Case $i = 1$ is evident from the definition of $\delta$ since $w_1=b_1^{d-1}$ and $\delta(q_{j-1}^1, b_1)=q_j^1$ for $j \in \{1,...d-1\}$.\newline Now let us assume that the result holds for $i - 1 < k$. From the induction hypothesis we know, that there exists a path 
$(Q_0^{i-1}, Q_1^{i-1},Q_2^{i-1}, ... ,Q_{d^{i-1} - 1}^{i-1} )$ whose transitions are labelled with consecutive letters of word $w_{i-1}$. Let $l \in \{0,... , d-1\}$ and notice that for every $r\in \{0, ..., d^{i-1} - 1\}$ we have $Q_{l\cdot d^{i-1} + r}^i = \{q_l^i\} \cup Q_r^{i-1}$ and $\delta$ on any of letters $b_1, ... , b_{i-1}$ maps $q_l^i$ to itself, so for each $l$ there exist a path $(Q_{l\cdot d^{i-1}}^i, Q_{l\cdot d^{i-1} + 1}^i,Q_{l\cdot d^{i-1} + 2}^i, ... , Q_{(l+1)\cdot d^{i - 1} - 1}^i )$ in $\mathcal{P}(\mathcal{A}(n))$ (also labelled with letters of $w_{i-1}$). From the definition of $\delta$ we can see that for each $l \in \{1, ... , d-1\}$ it holds that $\tau(Q_{l \cdot d^{i - 1} - 1}^i, b_i) = Q_{l \cdot d^{i-1}}^i$. Using these two observations and the definition of $w_i$ we conclude that the lemma holds.
\end{proof}

\begin{lemma}
Automaton $\mathcal{A}_d(n)$ is carefully synchronizing and its carefully synchronizing word is $v = aw_kc_kw_{k-1}c_{k-1}...w_2c_2$.
\end{lemma}

\begin{proof}
We must show that $|\tau(Q, v)| = 1$. From the definition of $\delta$ we see that $\tau(Q,a) = Q_0^k$. From Lemma 1. we know that for every $l \in \{2, ..., k\}: \tau(Q_0^l,w_l) = Q_{d^l - 1}^l$. Also from the definition of $\delta$ we see that for every $l \in \{2, ..., k\}: \tau(Q_{d^l - 1}^l,c_l) = Q_0^{l-1}$. Joining those facts together we deduce that $\tau(Q, v) = \{q_0^1\}$. 
\end{proof}

\begin{lemma}
Let $v$ be as in Lemma 2. Then $|v| = \frac{1}{d-1}(d^{k+1} + (d-1)k - d^2)$ and $v$ is the shortest carefully synchronizing word for automaton $\mathcal{A}_d(n)$.
\end{lemma}
\begin{proof}
First we will show that $|v| = \frac{1}{d-1}(d^{k+1} + (d-1)k - d^2)$. It's obvious that $|w_ic_i| = d^i$, so $|w_kc_kw_{k-1}c_{k-1}...w_2c_2| = \sum\limits_{i=2}^k (d^i + 1) = \frac{1}{d-1}(d^{k+1} + (d-1)k - d^2 - d + 1)$. We leave that identity as a simple exercise for a reader.\newline\newline
It can be easily verified that $\tau(Q,a) = Q_0^k$ and $a$ is the only letter defined for all states. In order to prove minimality of $v$ it suffices to show that for each state $Q_s^k \subset Q$ in $\mathcal{P}(\mathcal{A}_d(n))$ there is only one transition that leads to a state $Q_{s'}^{k'} \subset Q$ that has not been visited yet. All other transitions are either not defined or lead to states visited earlier. We must investigate two cases:
\newline\newline
Case $Q' = Q_r^m$ for some $r,m \in \mathbb{N}$ and $r \neq d^m - 1$ \newline
From the definition of $\delta$, $\tau(Q', a) = Q_0^m$  (which was visited) and from Fact 6. $\tau(Q',c_j)$ is not defined for any $c_j$. From Lemma 2. it can be seen that for each $Q_r^m$ there exists a letter $b_k$ which leads to an unvisited state $Q_{r+1}^{m}$. In order to show that there exists only one such letter let us assume, that $r = (a_1, a_2, ..., a_{m - k + 1}, d-1, d-1, ... , d-1)_d$, $k \geq 1$ and $a_{m - k + 1} \neq d-1$. It is obvious from Fact 4. that for each $b_l$ such that $l > k$ the transition $\tau(Q', b_l)$ is not defined. If $l < k$ then it follows from Fact 5. that $\tau(Q', b_l)$ is not defined and the statement is true for that case. \newline\newline
Case $Q' = Q_r^m$ for some $r,m \in \mathbb{N}$ and $r = d^m - 1$\newline
From the definition of $\delta$ we see that $\tau(Q',a) = Q_0^m$. Furthermore $\tau(Q',b_j) = Q_0^m$ when $j > m$ otherwise when $j \leq m$, due to Fact 5, transitions are not defined. Notice that $\tau(Q',c_j) = Q_0^m$ for $j > m$. If $j < m$, then $\tau(Q',c_j)$ is not defined. Moreover $\tau(Q',c_m) = Q_0^{m-1}$. Since there is only one letter leading to an unvisited state and transitions under other letters are either undefined or their result is already visited state $Q_0^m$ of $\mathcal{P}(\mathcal{A}_d(n))$ statement holds for that case.\newline\newline
Having that we know by induction that $w$ is minimal and that ends the proof.
\end{proof}
Following theorem is immediate from Lemma 4.
\begin{theorem}
Let $n = d\cdot k$, $k \in \mathbb{N}$. The shortest carefully synchronizing word for $\mathcal{A}_d(n)$ has length $O(d^\frac{n}{d})$. 
\end{theorem}

Using that theorem we can simply reproduce result obtained by Martyugin [8] as follows:
\begin{corollary}
If $n > 3$ then there exist a PFA with $n$ states and minimal carefully synchronizing word of length $O(3^\frac{n}{3})$.
\end{corollary}
\begin{proof}
We construct automaton $\mathcal{A}_3(m)$ with $m = n - (n\mod 3)$ and denote rest of states as $Q'$. Now we can add a letter to the automaton, say $d$, and add a transition over that letter to $\tau$, resulting with $\tau'$, such that it acts like identity on $Q \setminus Q'$ and $\tau'(Q', d) \in Q$.
\end{proof}

\section{Further improvements}
Define $\sigma_a$ a relation on the set of states $Q$ for an automaton $\mathcal{A}$ and a given letter $a \in \Sigma$ such that $q_1 \sigma_a q_2$ if, and only if $\delta(q_1,a) = \delta(q_2,a)$. It is obvious that, for any $a \in \Sigma$, $\sigma_a$ is an equivalence relation on the set of states. We also define $\sigma_a$-\textit{transversal} as $Q' \in Q$ such that each equivalence class has at most one representative in $Q'$. Let $Q_1, ..., Q_l$ be equivalence classes of $\sigma_a$ on $Q$. Finally we say that letter $b \in \Sigma$ is $\sigma_a$-\textit{preserving} with respect to $\tau: 2^Q \times \Sigma \rightarrow 2^Q$ if for any $\sigma_a$-transversal $Q' = \{q_{i_1}, ... ,q_{i_k} \}$, such that lower index of $q_i$ corresponds to $i$-th equivalence class, $\tau(Q',b) = \{q'_{i_1}, ... ,q'_{i_k}\}$ such that lower index of $q'_i$ corresponds to $i$-th equivalence class. It can be easily seen that letter $a$ in automaton $\mathcal{A}_d(n)$ defines $\sigma_a$ on the $Q$, resulting with partition of it on $k$ pairwise disjunctive sets, and letters $b_i$ for $i = 1, ..., k$ are $\sigma_a$-preserving. In Section 2. we defined the automaton that first creates $\sigma_a$ on $Q$ and then traverses some of transversals of that relation. Specifically, after applying letter $c_i$ on $A_d$ we reducing number of equivalence classes possible to traverse by one. It is natural question to ask if we can traverse more transversals than we have shown in Section 2. We give universal construction that can be used to improve the lower bound obtained by Martyugin.\newline
Main idea is to immediately reduce $Q$ to $k$ equivalence classes and then treat those classes as states of some synchronizable DFA or carefully synchronizable PFA. First we define construction sufficient to construct carefully synchronizing automaton with long shortest carefully synchronizing word for any given synchronizable DFA or carefully synchronizable PFA and next we apply that construction to \v{C}ern\'{y} automata $\mathcal{C}_n$ in order to give an upper bound for shortest carefully synchronizing word for such created automaton.\newline Let $\mathcal{B} = (S, \Delta, \gamma)$ be a finite automaton. Let $S = \{q_1, ..., q_k\}$ and $\Delta = \{c_1, ..., c_s\}$. We define PFA $\mathcal{A}_d(\mathcal{B}) = (Q, \Sigma, \delta)$ as follows:

\begin{itemize}
  \item $\Sigma = \{a, b_1, b_2, ... b_k, c_1, c_2, ... , c_l\}$
  \item $Q_i = \{q_0^i, q_1^i, ... ,  q_{d-1}^i\} $
  \item $Q = \bigcup\limits_{i = 1}^k Q_i$
\end{itemize} 
Let $i \in \{1,..., k\}$, $l \in \mathbb{N}$. we define partial transition function $\delta:{Q \times \Sigma}\rightarrow{Q}$ for $\mathcal{A}_d(\mathcal{B})$ as:
\begin{itemize}
\item $\delta(q_j^i, a) = q_0^i$, $j \in \{0,1, ... , d-1\}$
\item $\delta(q_{j-1}^i, b_i) = q_{j}^i$, $j \in \{1,2, ... , d-1\}$
\item $\delta(q_j^i, b_l) = q_j^i$, $i > l$
\item $\delta(q_{d-1}^i, b_l) = q_0^i$, $i < l$
\item $\delta(q_{d-1}^i, c_l) = q_0^j$, for all $i,l$ such that $\gamma(q_i, c_l) = q_j$ 
\end{itemize}

We start with following simple observations.

\begin{fact}
Let $\sigma_a$ be defined as above on $Q$. Letters $a, b_1, ..., b_k \in \Sigma$ are $\sigma_a$-preserving.
\end{fact}

Before moving further we prove following lemma.

\begin{lemma}
Let $\mathcal{P}(\mathcal{A}_d(\mathcal{B})) = (2^{Q}, \Sigma, \tau)$ be a power automaton for automaton $\mathcal{A}_d(\mathcal{B})$. Let $\{i_1, i_2, ... , i_s\} \subset \{1, ..., k\}$. Let also  $Q_0 = \{q_0^{i_1},q_0^{i_2},... ,q_0^{i_s}\}$ and $Q_{3^s - 1} = \{q_{d-1}^{i_1},q_{d-1}^{i_2},... ,q_{d-1}^{i_s}\}$ . Then the shortest path $p$ from $Q_0$ to $Q_{d^s - 1}$ in $\mathcal{P}(\mathcal{A}_d(\mathcal{B}))$ is of length $d^s - 1$.
\end{lemma}

\begin{proof}
Since showing that path $p$ of desired length exists is similar to the proof of Lemma 1. and the word which traverses all sets on $p$ is analogous to the word in the proof of Lemma 1. we omit that part of proof and focus on proving that at any point on $p$ there exist only one transition to state not visited before in order to show minimality of $p$.\newline
Notice that we can treat any $Q_r$ on $p$ as $d$-ary representation of $r$ just like in Lemma 3, the only difference is that we omit "empty spots" in $Q_r$. Another similarity is that after $Q_{r+1}$ is directly after $Q_r$ on $p$. Assume that $r = (a_1, a_2, ..., a_{s - k + 1}, d-1, d-1, ... , d-1)_d$, $k \in \{1, ..., s\}$ and $a_{s - k + 1} \neq d-1$.  It is obvious from Fact 4. that for each $b_l$ such that $l > k$ the transition $\tau(Q_r, b_l)$ is not defined. If $l < k$ we must investigate two cases. If $q_{d-1}^l \in Q_r$ then from Fact 5. $\tau(Q_r, b_l)$ is not defined. Else notice from definition of $\delta'$ that transition $\tau(Q_r, b_l)$ zeros some of positions younger than $s - k + 1$ and maps older positions and position $s - k + 1$ to itself, so the result of that transition is such $Q_{r'}$ that $r > r'$ which was visited earlier since all numbers between $0$ and $r$ must be on $p$. None of letters $c_i$ is defined. That concludes the proof.
\end{proof}

Having that construction we may prove main theorem of that section.

\begin{theorem}
Let $d > 1$. $\mathcal{B}$ is a synchronizing DFA(carefully synchronizing PFA) if, and only if $\mathcal{A}_d(\mathcal{B})$ is carefully synchronizing.
\end{theorem}

\begin{proof}
Let $\mathcal{P(B)} = (2^S, \Delta, \rho)$ be a power automaton for automaton $\mathcal{B}$ and $\mathcal{P}(\mathcal{A}_d(\mathcal{B})) = (2^{Q}, \Sigma, \tau)$ be a power automaton for automaton $\mathcal{A}_d(\mathcal{B})$. Fix $d > 1$. First we prove right implication. Since $\mathcal{B}$ is (carefully) synchronizing there exist a (carefully) synchronizing word $w = c_{k_1}...c_{k_s} \in \Sigma^\ast$. We now construct carefully synchronizing word $w'$ for $\mathcal{A}_d(\mathcal{B})$. Let $w' = \epsilon$. Since only letter $a$ is defined for all states we append letter $a$ to $w'$. Notice that $\tau(Q, a) = Q_0^k$. From Lemma 1. We know that there exist $u_0 \in \Sigma^\ast$ of length $d^k-1$ such that $\tau(Q_0^k) = Q_{d^k-1}^k$. Notice that for each $i \in \{1, ..., k\}$ $\tau(Q_{d^k-1}^k, b_i) = Q_0^k$ or is undefined. Let $P = \{q_{i_1}, ... ,q_{i_n} \}$. It is easy to observe that if $\rho(P, c_{k_i}) = P'$ such that $P' = \{q_{j_1}, ... ,q_{j_m}\}$, then $\tau(\{q_{d-1}^{i_1}, ... , q_{d-1}^{i_n},\},c_{k_i}) = \{q_0^{j_1}, ... , q_0^{j_m}\} = Q_{c_k}$. Any such $Q_{c_k}$ is subset of $Q_0^k$, so from Lemma 4. there exist word $w_k$ (of length $3^{|Q_{c_k}|} - 1$), such that $\tau(Q_{c_k}, w_k) = \{q_{d-1}^{j_1}, ... , q_{d-1}^{j_m}\} $. From that it is easy to notice that word $w' = aw_kc_{k_1}w_{k_1}...c_{k_{s-1}}w_{k_{s-1}}c_{k_s}$ carefully synchronizes $\mathcal{A}_d(\mathcal{B})$.\newline In order to prove left implication suppose, for the sake of contradiction, that there exist non-synchronizable DFA (non-carefully synchronizable PFA) $\mathcal{B}$ such that its $\mathcal{A}_d(\mathcal{B})$ is carefully synchronizable. Any (carefully) synchronizing word $w$ for $\mathcal{A}_d(\mathcal{B})$ must start with letter $a \in \Sigma$, which defines $\sigma_a$ on $Q$. Notice that for any $\sigma_a$-transversal $Q' \neq \{q_{d-1}^{i_1}, ... , q_{d-1}^{i_l}\}$, $\tau(Q', c_i)$ is not defined for $i = 1, ..., s$ and, due to Fact 7, there is no letter that can change traversed equivalence classes. That leads to contradiction, since $\mathcal{B}$ must be (carefully) synchronizing so to $\mathcal{A}_d(\mathcal{B})$ be carefully synchronizing.
\end{proof}

\begin{corollary}
If $\mathcal{B}$ is a synchronizing DFA(carefully synchronizing PFA) then the shortest carefully synchronizing word for $\mathcal{A}_d(\mathcal{B})$ is $\Omega(d^{n/d})$
\end{corollary}

\begin{proof}
Since $|w_k| = d^k-1$ and it labels the shortest path from $Q_0^k$ to $Q_{d^k-1}^k$ we conclude corollary holds.
\end{proof}

Now we are ready to bound shortest carefully synchronizing word for $\mathcal{A}_d(\mathcal{C}_n)$. Let $\mathcal{C}_n = (S, \Delta, \gamma)$ be a DFA such that $S = \{q_0, ..., q_{n-1}\}$, $\Delta = \{c_1, c_2\}$ and $\gamma$ is defined as follows:
\begin{itemize}
    \item $\gamma(q_0,c_1) = q_1$
    \item $\gamma(q_m,c_1) = q_m$ for $m \in \{1, ..., n-1\}$
    \item $\gamma(q_m,c_2) =  q_{m+1 (mod\;n)}$

\end{itemize}
Despite the shortest synchronizing word for $\mathcal{C}_n$ is $(c_1c_2^{n-1})^{n-2}c_1$ of length $(n-1)^2$ [1], we find another synchronizing word, more appropriate word to bound $d(\mathcal{A}_d(\mathcal{C}_n))$. 
\begin{lemma}
If $n > 2$ is even, then word $w_1 = (c_1c_2^2)^\frac{n}{2}(c_1c_2^{n-1})^{n-3}c_1$ synchronizes $\mathcal{C}_n$, otherwise $w_2 = (c_1c_2^2)^\frac{n+1}{2}(c_1c_2^{n-1})^{n-4}c_1$ synchronizes $\mathcal{C}_n$.
\end{lemma}
\begin{proof}
We prove for even $n$ since proof for odd $n$ is similar. Let $\mathcal{P}(\mathcal{C}_n) = (2^S, \Delta, \rho)$ be a power automaton for automaton $\mathcal{C}_n$. Denote $c_1c_2^2 = u$ and $c_1c_2^{n-1} = v$.  It is easy to check by induction on $k$ that if $k \leq n/2$, then $|\rho(S,u^k)| = \{q_0, q_2, ...,q_{2k-2}, q_{2k}, q_{2k+1}, ..., q_{n-2} \}$. So $\rho(S, u^\frac{n}{2}) = \{q_0, q_2, ...,q_{n-2}\} = S_{\frac{n}{2}}$.  Now consider action of $v^2$ on set $S_{\frac{n}{2}}$. We will proof by induction on $k$ that if $k < \frac{n}{2} - 1$, then $\rho(S_{\frac{n}{2}}, v^{2k}) = \{q_0, q_2, ..., q_{n -2k -2}\}$. If $k = 1$ then, from definition of $\gamma$, $\rho(S_{\frac{n}{2}}, c_1) = \{q_1, q_2, ..., q_{n-2}\}$. So it is easily seen that $\rho(S_{\frac{n}{2}}, c_1c_2^{n-1}c_1) = \{q_1, q_3, ..., q_{n-3}\}$, and $\rho(S_{\frac{n}{2}}, c_1c_2^{n-1}c_1c_2^{n-1}) = \{q_0, q_2, ..., q_{n-4}\}$. Assume that lemma holds for every $i < k$, then $\rho(S_{\frac{n}{2}}, v^{2k}) = \{q_0, q_2, ..., q_{n -2k - 2}\}$. Similarly as in $k=1$ case, applying word $v^2$ results with $\{q_0, q_2, ..., q_{n -2k - 4}\}$ so the statement holds. That implies $\rho(S_{\frac{n}{2}}, (c_1c_2^{n-1})^{n-4}) = \{q_0, q_2\}$. Notice that $\rho(\{q_0, q_2\}, c_1c_2^{n-1}c_1) = \{q_1\}$ and that ends proof.
\end{proof}

Having that we prove following theorem.

\begin{theorem}
If $n$ is even, then $d(\mathcal{A}_d(\mathcal{C}_n)) \leq  \frac{1}{d-1}[d^{n+1} + 2 \cdot d^n + (n-4)\cdot d^{\frac{n}{2} + 1} + (n-1)(d^{\frac{n}{2}} - d^2 -d^3) - 1]$. Otherwise $d(\mathcal{A}_d(\mathcal{C}_n)) \leq \frac{1}{d-1}[d^{n+1} + 2\cdot d^n + (2n-5)\cdot d^\frac{n}{2} - (n-1)(d+1)d^2] - (n-1)\cdot d^{\frac{n} + 1} + 2$.
\end{theorem}

\begin{proof}
Let $n$ be even. We construct for a given automaton a reset word of desired length. From Lemma 5 we know that $(c_1c_2^2)^{n/2}(c_1c_2^{n-1})^{n-3}c_1$, so we know that there exist carefully synchronizing word $w$ for $\mathcal{A}_d(\mathcal{C}_n)$ of form $aw_n \cdot \prod\limits_{i=1}^{\frac{n}{2}} c_1u_1^ic_2u_2^ic_2u_3^i \cdot (\prod\limits_{i=1}^{n-3} c_1w_1^ic_2w_2^i ... c_2w_{n-1}^i )\cdot c_1$. Denote $\prod\limits_{i=1}^{\frac{n}{2}} c_1u_1^ic_2u_2^ic_2u_3^i = v_1$ and $\prod\limits_{i=1}^{n-3} c_1w_1^ic_2w_2^i ... c_2w_{n-1}^i = v_2$. It is obvious that $|aw_n| = 3^n$. Now we calculate $|v_1|$. It is easy to notice that after applying letter $c_1$ number of equivalence classes traversed by $u_i^j$ reduces by one, so from Lemma 4. $|v_1| = \sum\limits_{i=1}^{\frac{n}{2} - 1} 3\cdot d^{n-i}$. Consider $|v_2|$. From the proof of Lemma 5 we can deduce that we reduce number of equivalence classes by one after two letters $c_1$ in $w_2$. Thus, from Lemma 4 we obtain $|v_2| =  \sum\limits_{i=\frac{n}{2}}^{n - 2} ((2n-2) \cdot d^{n - i}) - (n-1)\cdot d^{\frac{n}{2}} - (n-1) \cdot d^2 + 1$. Because we do not traverse all $(2n-2)$ equivalence classes when $i = n/2$ and  $i = n-2$ (see proof of Lemma 5), we substract $(n-1)\cdot d^{\frac{n}{2}} + (n-1) \cdot d^2$. After simple calculations we obtain $|w| = \frac{1}{d-1}[d^{n+1} + 2 \cdot d^n + (n-4)\cdot d^{\frac{n}{2} + 1} + (n-1)(d^{\frac{n}{2}} - d^2 -d^3) - 1]$. Similar analysis of word $w$ when $n$ is odd results with $|w| = \frac{1}{d-1}[d^{n+1} + 2\cdot d^n + (2n-5)\cdot d^\frac{n}{2} - (n-1)(d+1)d^2] - (n-1)\cdot d^{\frac{n} + 1} + 2$.
\end{proof}
 We are now able to formulate following corollary.
 \begin{corollary}
Let $d > 1$, $n > 2$. Then $d(\mathcal{A}_d(\mathcal{C}_n)) \in O(d^n + n\cdot d^\frac{n}{2})$ and  $d(\mathcal{A}_d(\mathcal{C}_n)) \in \Omega(d^n )$.
 \end{corollary}

\end{document}